\documentclass[12pt, a4paper]{article}

\usepackage{amsmath}
\usepackage{amssymb}
\usepackage{enumerate}
\usepackage{microtype}
\usepackage{url}
\usepackage{graphicx}
\usepackage{subfig}
\DeclareMathOperator{\tw}{tw}

\newtheorem{theorem}{Theorem}

\newtheorem{lemma}{Lemma}
\newtheorem{remark}{Remark}

\newcommand*{\qed}{\hbox{}\hfill$\Box$}

\newenvironment{proof}{\par\noindent\textbf{Proof.}}{\qed\medskip}
\begin{document}

\begin{center}
  {\LARGE Tree-width of hypergraphs and surface duality\footnote{This
      research is supported by the french ANR project DORSO.}}
  \vspace{1.5\baselineskip}

  {\large Fr\'ed\'eric Mazoit\footnote{Email:
      \url{Frederic.Mazoit@labri.fr}}} \vspace{\baselineskip}

  \emph{\small LaBRI, Universit\'e de Bordeaux,}

  \emph{351 cours de la Lib\'eration F-33405 Talence cedex, France}
\end{center}

\hrule\medskip

\noindent\textbf{Abstract}

In Graph Minors III, Robertson and Seymour write:``It seems that the
tree-width of a planar graph and the tree-width of its geometric dual
are approximately equal --- indeed, we have convinced ourselves that
they differ by at most one.'' They never gave a proof of this.  In
this paper, we prove a generalisation of this statement to embedding
of hypergraphs on general surfaces, and we prove that our bound is
tight.\bigskip

\medskip

\hrule

\section{Introduction}
\label{sec:intro}

Tree-width is a graph parameter which was first defined by
Halin~\cite{Ha76a}, and which has been rediscovered many times
(see~\cite{ArPr89a, RoSe84a}).  In~\cite{ArPr89a}, Arnborg and
Proskurovski introduced a general framework to efficiently solve
NP-complete problems when restricted to graphs of bounded tree-width.
Courcelle~\cite{Co90a} extended this framework by showing that any
problem expressible in a certain logic can be efficiently solved for
graphs of bounded tree-width.  Tree-width thus seems to be a good
``complexity measure'' for graphs.

Given an embedding $\Gamma$ of a graph in a surface, it is easy to
obtain the dual embedding $\Gamma^*$: just put a vertex in each face
and for every edge $e$ separating the faces $f$ and $g$, add a dual
edge $fg$.  One could thus expect that $\Gamma$ and $\Gamma^*$ have
the same ``complexity'', and indeed in~\cite{RoSe84a}, Robertson and
Seymour claimed that for a plane embedding $\Gamma$, $\tw(\Gamma)$ and
$\tw(\Gamma^*)$ differ by at most one.

In an unpublished paper, Lapoire~\cite{La96b} gave a more general
statement about embeddings of hypergraphs on orientable surfaces.
Nevertheless, his proof was rather long and technical.  Later,
Bouchitt\'e et al.  and Mazoit~\cite{BoMaTo03a, Ma04a} gave easier
proofs for plane graphs.  Here we give a proof that Lapoire's claim is
valid for general surfaces\footnote{This result also appears as an
  extended abstract in~\cite{Ma09b}.  Unfortunately although the
  general scheme of the proof is the same, some definitions are wrong
  and we could not obtain a valid proof with them.}:
\begin{theorem}\label{thm}
  For any 2-cell embedding $\Lambda$ of a hypergraph on a surface
  $\Sigma$,
  \begin{equation*}
    \tw(\Lambda^*)\leq\max\{\tw(\Lambda)+1+k_\Sigma, \alpha_{\Lambda^*}-1\}.
  \end{equation*}
\end{theorem}
Here $\alpha_{\Lambda^*}$ is the maximum size of an edge of
$\Lambda^*$ and $k_\Sigma$ is the Euler genus of $\Sigma$.  \medskip

In Section~\ref{sec:preliminaries}, we give the basic definitions.
Section~\ref{sec:upper-bound} is devoted to the proof of
Theorem~\ref{thm} while in Section~\ref{sec:examples} we give examples
of embeddings which match the bound of this theorem.

\section{Preliminaries}
\label{sec:preliminaries}

A \emph{tree-decomposition} of a hypergraph $H$ is a pair $\mathcal{T}
= (T, (X_v)_{v\in V_T})$ with $T$ a tree and $(X_v)_{v\in V_T}$ a
family of subsets of vertices of $H$ called \emph{bags} such that
every every edge of $H$ is contained in at least one bag of
$\mathcal{T}$, and for every vertex $v\in V_H$, the vertices of $T$
whose bag contain $v$ induces a non-empty sub-tree of $T$.  The
\emph{width} of $\mathcal{T}$ is $\tw(\mathcal{T}) =
\max\{|X_t|-1\;;\; t\in V_T\}$ and the \emph{tree-width} $\tw(H)$ of
$H$ is the minimum width of one of its tree-decompositions.

A \emph{surface} is a connected compact 2-manifold without boundaries.
Oriented surfaces can be obtained by adding ``handles'' to the sphere,
and non-orientable surfaces, by adding ``crosscaps'' to the sphere.
The \emph{Euler genus} $k_\Sigma$ of a surface $\Sigma$ (or just
\emph{genus}) is twice the number of handles if $\Sigma$ is
orientable, and is the number of crosscaps otherwise.

We denote by $\overline X$ the closure of a subset $X$ of $\Sigma$.
We say that two disjoint subsets $X$ and $Y$ of $\Sigma$ are
\emph{incident} if $X\cap\overline Y$ or $Y\cap \overline X$ is
non-empty.  Since we consider finite graphs and hypergraphs, we can
assume that the curves involved in the embeddings are not completely
wild and are, say, piecewise linear.  This implies that connectivity
and arc-connectivity coincide.  An \emph{open curve} is a subset of
$\Sigma$ which is homeomorphic to $]0, 1[$.  An open curve whose
closure is homeomorphic to the 1-sphere $S^1$ is a \emph{loop} and is
a \emph{strait edge} otherwise.  A connected subset $X$ of $\Sigma$ is
a \emph{star} if it contains a point $v_X$ called its \emph{centre}
such that $X\setminus\{v_X\}$ is a union of pairwise disjoint strait
edges called \emph{half edges}.  Note that an open curve is also a
star.  Let $X$ be an open curve or a star.  The elements of
$\overline{X}\setminus X$ are the \emph{ends} of $X$.

An \emph{embedding of a hypergraph on a surface $\Sigma$} is a pair
$\Lambda=(V, E)$ in which $V$ is a finite subset of $\Sigma$ whose
elements are the \emph{vertices} of the embedding, and $E$ is a finite
set of pairwise disjoint stars called \emph{(hyper)edges}.  Edges
contain no vertex and their ends are vertices.  Such an embedding
naturally corresponds to an abstract hypergraph $H$.  We say that
$\Lambda$ is an \emph{embedding of $H$}.  An \emph{embedding of a
  graph on $\Sigma$} is an embedding of a hypergraph whose edges are
strait edges and loops.  Let $\Lambda$ be an embedding of a hypergraph
on $\Sigma$.  We denote by $V_\Lambda$ the vertex set of $\Lambda$ and
by $E_\Lambda$ the edge set of $\Lambda$.  Let $V_{E_\Lambda}$ contain
the centre of all the edges and let $L_\Lambda$ contain all the half
edges.  Then $(V_\Lambda\cup V_{E_\Lambda}, L_\Lambda)$ is an
embedding of a bipartite graph on $\Sigma$ which is the
\emph{incidence embedding of $\Lambda$}.  We denote embeddings of
graphs with the Greek letters $\Gamma$ and $\Pi$ and embeddings of
hypergraphs with the Greek letter $\Lambda$.  We only consider
embedding of graphs and hypergraphs up to homeomorphisms.  Since
embeddings of hypergraphs naturally have abstract counterparts, we
apply graph theoretic notions to them without further notice.  For
example, we may consider an edge $e$ as a subset of $\Sigma$ or as a
set of vertices.  We also consider embeddings of hypergraphs on
$\Sigma$ as subsets of $\Sigma$.  \medskip

A \emph{face} of an embedding $\Lambda$ is a component of
$\Sigma\setminus \Lambda$.  We denote by $F_\Lambda$ the set of faces
of $\Lambda$.  An embedded hypergraph is \emph{2-cell} if all its
faces are homeomorphic to open discs.  Let $\Gamma$ be a 2-cell
embedding of a graph on a surface $\Sigma$.  Euler's formula links the
number of vertices, edges and faces of $\Gamma$ and the genus of the
surface:
\begin{equation*}
  |V_\Gamma|-|E_\Gamma|+|F_\Gamma| = 2-k_\Sigma.
\end{equation*}

We now let $\Lambda$ be a 2-cell embedding of a hypergraph on a
surface $\Sigma$.  The \emph{dual} of $\Lambda$ is the embedding
$\Lambda^*$ such that:
\begin{figure}[htbp]
  \centering
  \includegraphics{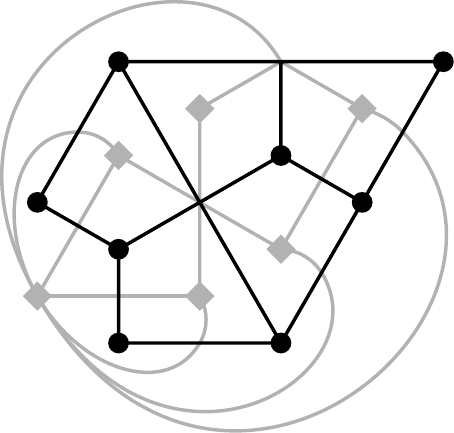}\qquad\qquad
  \includegraphics{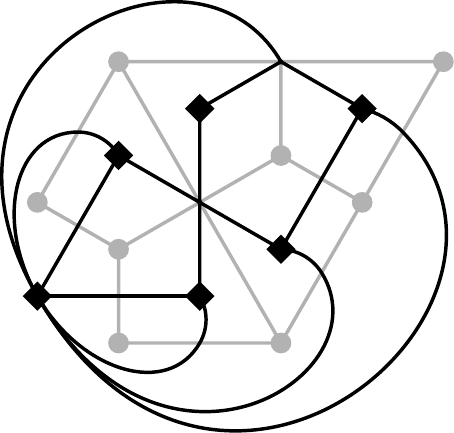}

  \caption{A planar hypergraph and its dual.}
  \label{fig:example}
\end{figure}
\begin{enumerate}[i.]
\item Every vertex of $\Lambda^*$ belongs to a face of $\Lambda$ and
  every face of $\Lambda$ contains exactly one vertex of $\Lambda^*$;
\item For every edge $e$ of $\Lambda$, there exists a dual edge $e^*$
  sharing its centre, and every edge of $\Lambda^*$ corresponds to an
  edge of $\Lambda$.
\item For every edge $e$ of $\Lambda$ with centre $v_e$, the half
  edges of $e$ and $e^*$ around $v_e$ alternate in their cyclic order.
\end{enumerate}
Note that the construction does not need $\Lambda$ to be 2-cell but if
not, $\Lambda^*$ is not unique and $(\Lambda^*)^*$ need not be
$\Lambda$.

\section{The upper bound}\label{sec:upper-bound}
Since Theorem~\ref{thm} is about 2-cell embedings, and since the
theorem is trivial for edge-less embeddings, we always consider
connected embeddings and hypergraphs with at least one edge.

The \emph{border} of a partition $\mu$ of $E_H$ is the set
$\delta_H(\mu)$ of vertices which are incident with edges in at least
two parts of $\mu$, and the \emph{border} of $A\subseteq E_H$ is
$\delta_H(A)=\delta_H(\{A, E_H\setminus A\})$.  A \emph{partitioning
  tree} of $H$ is a tree $T$ whose leaves are bijectively labelled by
edges of $H$.  Removing an internal node $v$ of $T$ results in a
partition of the leaves of $T$ and thus in a \emph{node-partition}
$\lambda_v$ of $E_H$.  Removing an edge $e$ of $T$ results in a
bipartition of the leaves of $T$ and thus in an \emph{edge-partition}
$\lambda_e$ of $E_H$.

\begin{lemma}
  Let $H$ be a connected hypergraph with at least one edge.  Let $T$
  be a partitioning tree of $H$.  Labelling each internal node $v$ of
  $T$ with $\delta_H(\lambda_v)$ turns $T$ into a tree-decomposition.
\end{lemma}
\begin{proof}
  By construction, every edge of $H$ is contained in a bag of $T$.
  Let $x\in V_H$.  Let $S$ be the set of leaves of $T$ whose label
  contain $x$, and let $T_x$ be the subtree of $T$ whose set of leaves
  is $S$.  Since $x$ is not isolated, $T_x$ contains at least one
  leaf.  Moreover, an internal bag of $T$ contains $x$ if and only if
  it separates two leaves $u$ and $v$ of $T$ whose edge label contain
  $x$.  Since those bags are precisely its internal bags, $T_x$ is
  precisely the subgraph induced by the vertices of $T$ whose bag
  contain $x$.
\end{proof}

The \emph{tree-width} of a partitioning tree is its \emph{tree-width},
seen as a tree-de\-com\-po\-si\-tion.

Let $\Lambda$ be a 2-cell embedding of a hypergraph on a surface
$\Sigma$.  If $T$ is a partitioning tree of $\Lambda$, then
\emph{dual} of $T$ is the partitioning tree $T^*$ of $\Lambda^*$
obtained by replacing in $T$ each label $e$ by the dual edge $e^*$.
\medskip

Given these definition, it is tempting to try to prove that for any
embedding $\Lambda$ of a hypergraph on $\Sigma$:
\begin{enumerate}[i.]
\item there always exists a partitioning tree $T$ such that
  $\tw(T)=\tw(\Lambda)$;
\item for any partitioning tree $T$, $\tw(T^*)\leq
  \max\{\tw(T)+1+k_\Sigma, \alpha_{\Lambda^*}-1\}$.
\end{enumerate}
The first item is true but we could not prove the second one.
However, we prove that both properties hold for a restricted class of
partitioning tree which we call p-trees.

\subsection{A sketch of the planar case}
Before we go on with the proof, we consider the planar case as it
contains most ideas.  The proof which we now sketch is based on the
proof in~\cite{Ma04a}.  Note that all definitions in this subsection
are local to this subsection.

Let $\Gamma$ be an embedding of a graph on the sphere $S^2$.
Moreover, let us suppose that $\Gamma$ has no bridge and no loop.  A
\emph{pretty curve} is a subset of $S^2$ which is homeomorphic to
$S^1$, which crosses $\Gamma$ only on vertices, and which never
``enters'' a face twice or more.  A \emph{$\Theta$-structure} is a
union of three curves $\rho_e\cup \rho_f\cup \rho_g$ such that
$\rho_e\cup \rho_f$, $\rho_f\cup \rho_g$ and $\rho_g\cup \rho_e$ are
all pretty curves.  Pretty curves induce bipartitions of $E_\Gamma$,
and $\Theta$-structures induce tripartitions of $E_\Gamma$.  A
partitioning tree is \emph{geometric} if all its node partitions come
from $\Theta$-structures.

Let $T$ and $T^*$ be dual geometric partitioning tree, and let $v$ be
a node of $T$.  We claim that the size of the dual bags $X_v$ and
$X^*_v$ differ by at most 1.  This is clearly true for leaf bags whose
size is either 1 or 2.  So let us suppose that $v$ is an internal node
and let $\Theta= \rho_e\cup \rho_f\cup \rho_g$ be a $\Theta$-structure
realising the node partition $\lambda_v$.  By construction $X_v$
contains all the vertices which belong to $\Theta$, and $X^*_v$
contains all the faces that $\Theta$ goes through.  Since
$\rho_e\cup\rho_f$ is a pretty curve which alternatively crosses
vertices and faces of $\Gamma$ and never enters the same face twice,
$|X_v\cap(\rho_e\cup\rho_f)|=|X_v^*\cap(\rho_e\cup\rho_f)|$.  The
difference between $|X_v|$ and $|X_v^*|$ thus comes from $\rho_g$.
But $\rho_g$ also alternatively crosses vertices and faces of $\Gamma$
without entering the same face twice.  This implies that difference
between $|X_v|$ and $|X_v^*|$ is at most 1.  Since this inequality
holds for any node, we have $\tw(T^*)\leq\tw(T)+1$.

To finish, we only need to prove that there exists a geometric
partitioning tree $T$ such that $\tw(\Gamma)=\tw(T)$.  To do this, we
apply an induction on planar hypergraphs.  Suppose that $\rho$ is a
pretty curve whose bipartition of $E_\Gamma$ is $\{A, B\}$.  Let $D_A$
and $D_B$ be the two components of $S^2\setminus\rho$.  If we remove
all the vertices and edges in $D_A$ and replace them by a star whose
set of ends if $\delta_\Gamma(A)$, be obtain a contracted hypergraph
$\Gamma_{/A}$.  Let $T_{/A}$ and $T_{/B}$ be geometric partitioning
trees of $\Gamma_{/A}$ and $\Gamma_{/B}$.  By removing from the
disjoint union $T_{/A}\cup T_{/B}$ the leaves labelled $e_A$ and $e_B$
and adding a new edge between their respective neighbours, we obtain a
geometric partitioning tree $T$ of $\Gamma$.  We show that
$\tw(T)=\max\{\tw(T_{/A}), \tw(T_{/B})\}$.  By induction
$\tw(\Gamma_{/A})=\tw(T_{/A})$ and $\tw(\Gamma_{/B})=\tw(T_{/B})$.
The result follows from the fact that it is always possible to find
$\rho$ such that $\tw(\Gamma)=\max\{\tw(\Gamma_{/A}),
\tw(\Gamma_{/B})\}$.  \medskip

Before we go on with the general case, let us make some comments.
\begin{itemize}
\item On higher genus surfaces, we can still describe separators in
  terms of curves on $\Sigma$ but as the genus increases, the number
  of curves involved increases and it quickly becomes too complex to
  control how the curves interact.  As a matter of fact we do not
  really care about curves.  $\Theta$-structures are only important
  because they cut the sphere in three connected regions.

  Indeed, we can prove that $|X^*_v|\leq |X_v|+1$ without considering
  curves as follows.  Let $\Theta$ be a $\Theta$-structure which
  realises $\lambda_v$, and let $D_A$, $D_B$ and $D_C$ be the
  components of $S^2\setminus\Theta$.  We contract all the vertices
  and faces which are contained in $D_A$ into a single vertex $v_A$.
  We do the same thing in $D_B$ and $D_C$ and to obtain two vertices
  $v_B$ and $v_C$.  We obtain a bipartite embedding $\Gamma_v$ whose
  set of vertices is $\{v_A, v_B, v_C\}\cup X_v$ and whose set of
  faces is $X^*_v$.  Euler's formula thus implies that
  $(|X_v|+3)+|X^*_v|-|E_{\Gamma_v}|=2$.  Since the faces of $\Gamma_v$
  are incident with at least 4 vertices, it is easy to prove that
  $|E_{\Gamma_v}|\geq 2|X^*_v|$.  The bound for the planar case
  follows.

  Although the contracting process becomes quite technical, this proof
  does generalise to higher genus surfaces.

\item Some loops and bridges are troublesome and must be taken care of
  separately.

  Indeed, let $e$ be a loop of $\Gamma$ which separates vertices of
  $\Gamma$, and let $v$ be an internal node of $T$ which is the
  neighbour of a leaf labelled by $e$.  Since any curve which isolates
  $e$ from $E\setminus\{e\}$ has to enter the end of $e$ twice, the
  node partition $\lambda_v$ cannot come from a $\Theta$-structure.
  The same kind of problem arises if $e$ is a separating bridge
  because any curve isolating $e$ must enter the same face twice.

  Let $e$ be a separating loop.  If we take a closer look, the bag
  $X_v$ is $\{e\}$ so the bag $X_v^*$ should be $\{e^*\}$.  For such
  internal nodes, we could just drop the condition $\lambda_v$ comes
  from a $\Theta$-structure and take any partition whose border is
  contained in $e$.  But this idea does not work.  For example, in
  figure~\ref{fig:troublesome-loop}, $e$ is a separating loop whose
  end is $v$.  The border of the partition $\bigl\{ \{a\}, \{b, c,
  d\}, \{e\}, \{f\}, \bigr\}$ is $\{v\}$ but the border of the dual
  partition is $\{F1, F2, F3\}$ whereas the dual of $e$ is the edge
  $\{F1, F2\}$.

  \begin{figure}[htbp]
    \centering
    \includegraphics{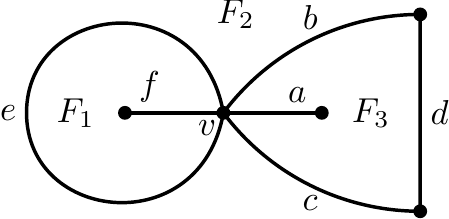}
    \caption{A troublesome loop.}
    \label{fig:troublesome-loop}
  \end{figure}

  To make this work, we cannot just take any partition whose border is
  a subset of $e$, we have to take a partition whose parts correspond
  to the connected components of $S^2\setminus \bar e$.

\item In the planar case, for our purpose, any two pretty curve which
  induce the same bipartition of $E_\Gamma$ are equivalent.  To avoid
  explicitly considering equivalent classes of pretty curves, the
  proof in~\cite{Ma04a} proceed as follows.  For each face $F\in
  F_\Gamma$, it put a vertex which is linked to all the vertices in
  $V_\Gamma$ which are incident to $F$.  The resulting embedding $\Pi$
  a \emph{radial embedding of $\Gamma$}.  Pretty curves then
  correspond to cycles of $\Pi$.

  In this paper, we do not explicitly realise partitions of $E_\Gamma$
  with curves on $\Sigma$ but with some disjoint connected subsets
  $\Sigma_A$, $\Sigma_B$, $\Sigma_C$ of $\Sigma$.  As for curves,
  there is not a single way to realise a node partition with such
  subsets and we use radial embeddings to avoid dealing with
  cumbersome equivalent classes.
\end{itemize}

\subsection{Partitioning trees}
Given a non-empty subset $A\subseteq E_H$, we define the
\emph{contracted hypergraph} $H_{/A}$ of $H$ as the hypergraph with
vertex set $\cup(E_H\setminus A)$ and with edge set $(E_H\setminus
A)\cup \{e_A\}$ in which $e_A=\delta_H(A)$ is a new hyperedge.  Let
$\{A, B\}$ be a non trivial bipartition of $E_H$ and $T_{/A}$ and
$T_{/B}$ be respectively partitioning trees of $H_{/A}$ and $H_{/B}$.
By removing from the disjoint union $T_{/A}\cup T_{/B}$ the leaves
labelled $e_A$ and $e_B$ and adding a new edge between their
respective neighbours, we obtain a partitioning tree $T$ which is the
\emph{merge of $T_{/A}$ and $T_{/B}$}.
\begin{lemma}\label{lem:merge_P-tree}
  Let $H$ be a connected hypergraph with at least one edge.  Let $\{A,
  B\}$ be a non trivial bipartition of $E_H$, and let $T_{/A}$ and
  $T_{/B}$ be partitioning trees of $H_{/A}$ and $H_{/B}$.  Then the
  merge $T$ of $T_{/A}$ and $T_{/B}$ is such that
  \[\tw(T)=\max\{\tw(T_{/A}), \tw(T_{/B})\}.\]
\end{lemma}
\begin{proof}
  Let $C\subseteq E_H$ be disjoint from $A$.  We claim that
  $\delta_H(C)$ and $\delta_{H_{/A}}(C)$ are equal.  Indeed, let
  $v\in\delta_{H_{/A}}(C)$.  By definition, there exists $e\in
  E_{H_{/A}}\setminus C$ and $f\in C$ which contain $v$.  If $e\neq
  e_A$, then $e\in E_H\setminus C$.  Otherwise, $e=e_A=\delta_H(A)$
  and there exists $e'\in A\subseteq E_H\setminus C$ which contains
  $v$.  In both cases, $v\in\delta_H(C)$.  Conversely, let
  $v\in\delta_H(C)$.  By definition, there exists $e\in E_H\setminus
  C$ and $f\in C$ which contain $v$.  If $e\notin A$, then $e\in
  E_{H_{/A}}\setminus C$.  Otherwise $v\in\delta_H(A)=e_A$.  In both
  cases $v\in\delta_{H_{/A}}(C)$.

  Let $u$ be an internal node of $T$.  By symmetry, we can suppose
  that $u$ belongs to $T_{/A}$.  The node-partition of $u$ in $T_{/A}$
  is ${\lambda_u}_{/A}=\{E_1\cup\{e_A\}, E_2, \dots, E_p\}$, and the
  node-partition of $u$ in $T$ is $\lambda_u=\{E_1\cup A, E_2, \dots,
  E_p\}$.  The above claim implies that $\delta_H(\lambda_u) =
  \delta_{H_{/A}}({\lambda_u}_{/A})$. The result follows.
\end{proof}

Lemma~\ref{lem:merge_P-tree} and the following folklore lemma are the
key tools to our proof of Theorem~\ref{th:existsT} that there always
exists a p-tree of optimal width.
\begin{lemma}\label{lem:1}
  Let $H$ be a hypergraph with at least one edge and no isolated
  vertices.  For any bipartition $\{A, B\}$ of $E_H$,
  \begin{equation*}
    \tw(H)\leq\max\{\tw(H_{/A}), \tw(H_{/B})\}.
  \end{equation*}
  If $\delta_H(\{A, B\})$ belongs to a bag of an optimal
  tree-decomposition of $H$, then
  \begin{equation*}
    \tw(H)=\max\{\tw(H_{/A}), \tw(H_{/B})\}.
  \end{equation*}
\end{lemma}
\begin{proof}
  Let $\mathcal{T}_{/A}=(T_{/A}, (X_v)_{v\in V_{T_{/A}}})$ and
  $\mathcal{T}_{/B}=(T_{/B}, (Y_v)_{v\in V_{T_{/B}}})$ be respective
  optimal tree-decompositions of $H_{/A}$ and $H_{/B}$.  Let $u$ be a
  vertex of $T_{/A}$ whose bag contain $e_{/A}$ and let $v$ be a
  vertex of $T_{/B}$ whose bag contain $e_{/B}$.  By adding an edge
  $uv$ to the disjoint union $T_{/A}\cup T_{/B}$, we obtain a
  tree-decomposition $\mathcal{T}$ of $H$ such that $\tw(\mathcal{T})
  = \max\{\tw(H_{/A}), \tw(H_{/B})\}$, which proves the first part of
  the lemma.

  Suppose now that $\delta_H(\{A, B\})$ belongs to the bag of a vertex
  $v$ of an optimal tree-decomposition $\mathcal{T}=(T, (Z_v)_{v\in
    V_T})$ of $H$.  By removing $V\setminus(\cup B)$ from the bags of
  $\mathcal{T}$, we obtain a tree-decomposition $\mathcal{T}_{/A}$ of
  $H_{/A}$ such that $\tw(\mathcal{T}_{/A})\leq\tw(H)$.  Similarly, we
  obtain a tree-decomposition $\mathcal{T}_{/B}$ of $H_{/B}$ such that
  $\tw(\mathcal{T}_{/B})\leq\tw(H)$.  The second part of the lemma
  follows.
\end{proof}

\subsection{Radial embeddings}
Note that, in this subsection, we do not require embedded hypergraphs
to be 2-cell but they must be connected and have at least one edge.

\begin{figure}[htbp]
  \centering

  \includegraphics{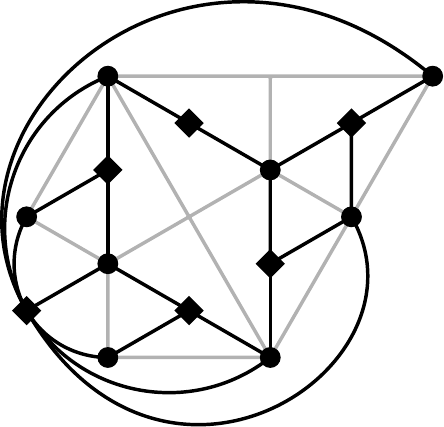}

  \caption{A radial embedding of the left example of
    Figure~\ref{fig:example}.}
  \label{fig:radial}
\end{figure}

Let $\Lambda$ be an embedding of a hypergraph on a surface $\Sigma$.
A \emph{radial embedding} of $\Lambda$ is an embedding $\Pi$ of a
bipartite graph on $\Sigma$ such that:
\begin{enumerate}[i.]
\item $\{V_\Lambda, V_\Pi\setminus V_\Lambda\}$ is a bipartition of
  $V_\Pi$, and $V_\Pi\setminus V_\Lambda$ contains exactly one vertex
  per face of $\Lambda$;
\item each edge of $\Lambda$ is contained in a face of $\Pi$ and each
  face of $\Pi$ contains exactly one edge of $\Lambda$.
\end{enumerate}

First radial embeddings do exist.
\begin{lemma}\label{lemma:radial}
  Every embedding $\Lambda$ of a connected hypergraph with at least
  one edge on a surface $\Sigma$ admits a radial embedding.
\end{lemma}
\begin{proof}
  The set $V_\Lambda$ being fixed, let us first choose one \emph{face
    vertex} per face of $\Lambda$ to get $V_\Pi\setminus V_\Lambda$.
  Let $(D_e)_{e\in E_\Lambda}$ be pairwise disjoint open discs such
  that each $D_e$ contains $e$.  Such discs can be obtained by
  ``thickening'' each edge a little.  We now continuously distort all
  the discs intersecting a given face so that they become incident
  with its corresponding ``face vertex''.

  The union of the boundaries of the discs $D_e$ correspond to the
  drawing of a bipartite graph $\Gamma$ that satisfies all the
  required condition except that some faces may be empty.  Indeed,
  suppose that $F\in F_\Lambda$ is homeomorphic to a disc and that
  $v\in V_\Lambda$ is incident with $F$.  Let $e_1$ and $e_2\in
  E_\Lambda$ bound an ``angle at $v$ in $F$''.  Between $e_1$ and
  $e_2$, there is an edge $f_1\in E_\Gamma$ which is in the boundary
  of $D_{e_1}$ and and edge $f_2\in E_\Gamma$ which is in the boundary
  of $D_{e_2}$.  The edges $f_1$ and $f_2$ bound an empty face.

  As long as $\Gamma$ contains an empty edge $F$, we remove any edge
  incident with $F$ to merge it with a neighbouring face and thus
  decrease the total number of empty faces of $\Gamma$.  In the end,
  we obtain a radial embedding $\Pi$ of $\Lambda$.
\end{proof}

If $\Lambda$ is a 2-cell embedding, then the radial embedding of
$\Lambda$ is unique and two distinct embeddings share the same radial
embedding if and only if they are dual embeddings.  But, as already
mentioned, we consider embeddings which are not 2-cell.  This implies
that a given embedding may have more than one radial embedding (see
Figure~\ref{fig:2-radial}).
\begin{figure}[htbp]
  \centering
  \includegraphics{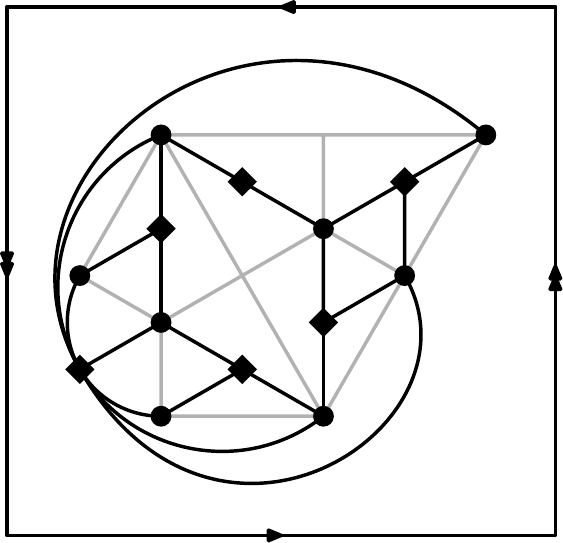}\qquad
  \includegraphics{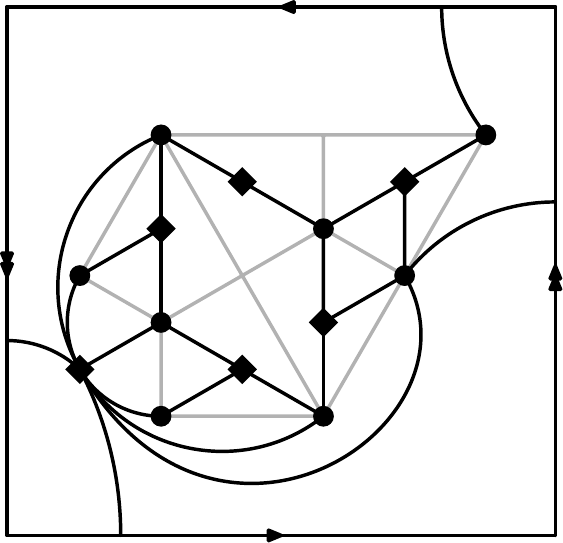}
  \caption{Two radial embeddings of a hypergraph on the projective
    plane.}
  \label{fig:2-radial}
\end{figure}

Let $\Lambda$ be an embedding of a hypergraph on a surface $\Sigma$,
and let $\Pi$ be a radial embedding of $\Lambda$.  We say that an edge
or a vertex of $\Pi$ is \emph{private to} a set $F$ of faces of $\Pi$
if all the faces it is incident to belong to $F$.

We now define several notions \emph{with respect to a radial embedding
  $\Pi$}.  Let $A$ be a set of edges of $\Lambda$.  We denote by
$A_\Pi$ the open set that contains all the faces of $\Pi$
corresponding to edges in $A$ together with the edges and vertices of
$\Pi$ which are private to these faces.  We say that $A$ is
\emph{$\Pi$-connected} if $A_\Pi$ is connected, and that a partition
of $E$ is \emph{$\Pi$-connected} if its parts are $\Pi$-connected.
Two edges $e$ and $f$ of $\Lambda$ are \emph{$\Pi$-adjacent} if $\{e,
f\}_\Pi$ is $\Pi$-connected.  An edge $e$ of $\Lambda$ is
\emph{troublesome} if the partition $\{e, E\setminus\{e\}\}$ is not
$\Pi$-connected.  The components of $\{e\}_\Pi$ then induce a
partition of $E_\Lambda\setminus\{e\}$.  Together with $\{e\}$, this
partition is the \emph{$e$-partition}.

If a vertex of $\Pi$ is private to a set of faces of $\Pi$, then so
are all its incident edges.  Thus if we denote by $G^\Pi$ the graph
whose vertices are the edges of $\Lambda$, and in which two vertices
are adjacent if their corresponding edges are $\Pi$-adjacent, then
$\Pi$-connected sets of edges of $\Lambda$ exactly correspond to
connected subgraphs of $G^\Pi$.

Let $A$ be a $\Pi$-connected set of edges of $\Lambda$.  Let us denote
$\widetilde{A}$ and $H$ the respective abstract counterparts of $A$
and $\Lambda$.  If we remove the edges and vertices of $\Lambda$ which
are contained in $A_\Pi$ and replace them by an edge $e_A$ whose set
of ends is $\delta_\Lambda(A)$ (which is possible because $A_\Pi$ is
connected), we obtain an embedding of a hypergraph whose abstract
counterpart is $H_{/\widetilde{A}}$.  We thus denote this new
embedding by $\Lambda_{/A}$.  By removing from $\Pi$ all the edges and
vertices which are contained in $A_\Pi$, we obtains the
\emph{contracted radial embedding} $\Pi_{/A}$ of $\Lambda_{/A}$.

\begin{remark}\label{rem}
  Let $A\subseteq E_\Lambda$ be $\Pi$-connected.  By construction of
  $\Pi_{/A}$,
  \begin{enumerate}[i.]
  \item a partition $\{E_1\cup A, E_2, \dots, E_l\}$ is
    $\Pi$-connected if and only if the partition $\{E_1\cup\{e_A\},
    E_2, \dots, E_l\}$ is $\Pi_{/A}$-connected;
  \item edge $e$ is troublesome in $\Lambda_{\/A}$ if and only if $e$
    is a troublesome edge in $\Lambda$ (and thus $e$ belongs to
    $\Lambda$);
  \item moreover, $\{\{e\}, E_2\cup \{e_A\}, E_3, \dots, E_p\}$ is the
    $e$-partition in $\Lambda_{/A}$ if and only if $\{\{e\}, E_2\cup
    A, E_3, \dots, E_p\}$ is the $e$-partition in $\Lambda$.
  \end{enumerate}
\end{remark}

\subsection{P-trees}

Let $\Lambda$ be an embedding of a connected hypergraph with at least
one edge on a surface $\Sigma$, and let $\Pi$ be a radial embedding of
$\Lambda$.  A \emph{p-tree} of $(\Lambda, \Pi)$ is a partitioning tree
$T$ of $\Lambda$ such that:
\begin{enumerate}[i.]
\item if a edge-partition $\lambda_e$ of $T$ is not $\Pi$-connected,
  then $e$ is incident with a leaf labelled by a troublesome edge.
\item if $v$ is an internal node of $T$ whose degree is not 3, then
  $v$ is a neighbour of a leaf labelled by a troublesome edge $e$ and
  $\lambda_v$ is the $e$-partition.
\end{enumerate}
Note that when no troublesome edge exist, all edge partitions are
$\Pi$-connected and all internal nodes have degree three.

Remark~\ref{rem} implies that:
\begin{lemma}\label{lem:2}
  Let $\Lambda$ be an embedding of a connected hypergraph with at
  least one edge on a surface $\Sigma$, let $\Pi$ be a radial
  embedding of $\Lambda$, and let $\{A, B\}$ be a $\Pi$-connected
  bipartition of $E_\Lambda$.  Let $T_{/A}$ and $T_{/B}$ be p-tree of
  $(\Lambda_{/A}, \Pi_{/A})$ and $(\Lambda_{/B}, \Pi_{/B})$
  respectively.  The merged partitioning-tree $T$ of $T_{/A}$ and
  $T_{/B}$ is a p-tree of $(\Lambda, \Pi)$.
\end{lemma}
\begin{proof}
  Let $e$ be an edge of $T$, and let $\lambda_e$ be its edge-partition
  in $T$.  If $e$ is the edge linking $T_{/A}$ and $T_{/B}$, then
  $\lambda_e=\{A, B\}$ is $\Pi$-connected.  Otherwise, by symmetry, we
  can suppose that, say, $e$ belong to $T_{/A}$.  By
  Remark~\ref{rem}.i., if edge-partition of $e$ in $\Lambda_{/A}$ is
  $\Pi_{/A}$-connected, then the edge-partition of $e$ in $\Lambda$ is
  $\Pi$-connected.  Otherwise, say, $e$ is incident with a leaf
  labelled by a troublesome edge in $\Lambda_{/A}$, and by
  Remark~\ref{rem}.ii., $e$ is incident with a leaf labelled by a
  troublesome edge in $\Lambda$.  Moreover,
  Remark~\ref{rem}.iii. directly implies that the node contition of
  p-trees is satisfied.
\end{proof}

We can now prove that
\begin{theorem}\label{th:existsT}
  Let $\Lambda$ be an embedding of a connected hypergraph with at
  least one edge on a surface $\Sigma$, and let $\Pi$ be a radial
  embedding for $\Lambda$.  There exists a p-tree $T$ of $(\Lambda,
  \Pi)$ such that $\tw(T)=\tw(\Lambda)$.
\end{theorem}
\begin{proof}
  We proceed by induction on $|V_\Lambda|+|E_\Lambda|$.  We say that
  an embedding $\Lambda$ is \emph{smaller} than another embedding
  $\Lambda'$ if
  $|V_\Lambda|+|E_\Lambda|<|V_{\Lambda'}|+|E_{\Lambda'}|$.  Let
  $\mathcal{T}$ be a tree-decomposition of $\Lambda$ of optimal width.
  We may assume that $T$ has no two neighbouring bags, one of which
  contains the other --- otherwise, contract the corresponding edge in
  $T$ and merge the bags.

  Suppose that we find a $\Pi$-connected bipartition $\{A, B\}$ of
  $E_\Lambda$ whose border is contained in a bag of $\mathcal{T}$, and
  such that $\Lambda_{/A}$ and $\Lambda_{/B}$ are smaller than
  $\Lambda$.  Such a partition is \emph{good}.  By induction, let
  $T_{/A}$ and $T_{/B}$ be p-trees of $(\Lambda_{/A}, \Pi_{/A})$ and
  $(\Lambda_{/B}, \Pi_{/B})$ of optimal width, and let $T$ be the
  merge of $T_{/A}$ and $T_{/B}$.  By Lemma~\ref{lem:2}, $T$ is a
  p-tree of $(\Lambda, \Pi)$.  By Lemma~\ref{lem:merge_P-tree}, its
  tree-width is $\max\{\tw(T_{/A}), \tw(T_{/B})\} =
  \max\{\tw(\Lambda_{/A}), \tw(\Lambda_{/B})\}$.  Since,
  $\delta_\Lambda(\{A, B\})$ is contained in a bag of $\mathcal{T}$,
  Lemma~\ref{lem:1} implies that
  $\tw(\Lambda)=\max\{\tw(\Lambda_{/A}), \tw(\Lambda_{/B})\}$, and
  thus, $\tw(T)=\tw(\Lambda)$.  We thus only have to find good
  partitions to complete the inductive step of our proof.

  Three cases arise:
  \begin{itemize}
  \item $\Lambda$ contains a troublesome edge $e$.

    If $e$ contains and separates all the other edges of $\Lambda$,
    then the partitioning star with one internal node is a p-tree of
    optimal width.  Otherwise, there exists a set of edges $A$ such
    that $A_\Pi$ is a component of $(E_\Lambda\setminus\{e\})_\Pi$ and
    $\Lambda_{/A}$ is smaller that $\Lambda$.  Since $e$, and thus
    $\delta_\Lambda(A)$, is contained in least one bag of
    $\mathcal{T}$ and since $\Lambda_{/(E_\Lambda\setminus A)}$ is
    also smaller than $\Lambda$ then $\{A, E_\Lambda\setminus A\}$ is
    good.

  \item $\mathcal{T}$ contains at least two nodes.

    In any tree-decomposition of $\Lambda$ with no bag being a subset
    of a neighbouring one, the intersection of two neighbouring bags
    is a separator of $\Lambda$.  There thus exists a separator $S$
    which is contained in a bag of $\mathcal{T}$.

    Let $C$ be a component of $\Lambda\setminus S$, and let $E_C$ be
    the sets of edges which are incident with vertices in $C$.  The
    set $E_C$ is $\Pi$-connected.  Let $\Pi_{E_1}$, \dots, $\Pi_{E_p}$
    be the components of $\Pi_{E\setminus E_C}$.  Since
    $S':=\delta_\Lambda(E_C)\subseteq S$ is a separator, then there
    exists a component $D$ of $\Lambda\setminus S'$ which is not $C$.
    The set of edges which are incident with $D$ is $\Pi$-connected,
    and is thus a subset of, say, $E_1$.  Since the sets
    $\overline{\Pi_{E_i}}$ ($2\leq i\leq p$) are incident with
    $\Pi_{E_C}$, then $\mu:=\{E_1, E\setminus E_1\}$ is
    $\Pi$-connected.  Since both $\Lambda_{/E_1}$ and
    $\Lambda_{/(E_\Lambda\setminus E_1)}$ contain fewer vertices than
    $\Lambda$ and are thus smaller than $\Lambda$, and since
    $\delta_\Lambda(\mu)\subseteq S'\subseteq S$ is contained in a bag
    of $\mathcal{T}$, then $\mu$ is good.

  \item $\mathcal{T}$ is the trivial decomposition of $\Lambda$ with
    one node.

    If $\Lambda$ contains at most three edges, then all partitioning
    trees are p-trees, so we can suppose that $\Lambda$ contains at
    least four edges.  Since no edge of $\Lambda$ is troublesome,
    $G^\Pi$ contains at least two vertices of degree at least 2.  And
    since $G^\Pi$ is connected, it contains at least two disjoint
    edges which can be extended in a non trivial bipartition of
    $G^\Pi$ in two connected sets.  This partition corresponds to a
    non trivial $\Pi$-connected partition $\mu:=\{A, B\}$ of
    $E_\Lambda$.  Since $\Lambda_{/A}$ and $\Lambda_{/B}$ contain
    fewer edges than $\Lambda$ and are thus smaller than $\Lambda$,
    then $\mu$ is good.
  \end{itemize}
\end{proof}

\subsection{P-trees and duality}

We now turn to the second step of our proof, that is, the tree-width
of a p-tree and that of its dual cannot differ too much.

Let us recall the strategy outlined in the sketch of the planar case.
We prove that if $T$ is a p-tree, then for every node $v$ of $T$, the
corresponding bags $X_v$ in $T$ and $X^*_v$ in $T^*$ have roughly the
same size.  To do so, if $v$ is an internal node of $T$, then we
construct a 2-cell embedding $\Lambda'$ with vertex set $X_v$ and with
face set $X^*_v$.  We then apply Euler's formula on the graph of the
incidence relation between vertices and edges of $\Lambda'$ to obtain
our bound.  It is easy to define an embedding with vertex set $X_v$.
Indeed, if the node partition of $v$ is $\{A, B, C\}$, then
$((\Lambda_{/A})_{/B})_{/C}$ will do.  The problem is that this
embedding is not 2-cell, and that its face set need not be $X^*_v$.
We thus have to be more careful when contracting to obtain such an
embedding.  The only problem is that we may have to consider more than
3 subsets of $\Sigma$ to realise $\lambda_v$.  For example, suppose
that $\Sigma$ is the torus and that $A_\Pi$ is a cylinder.  We may be
forced to replace $A_\Pi$ by 2 discs that fill the holes of
$\Sigma\setminus A_\Pi$ and consider both discs in our realisation of
$\lambda_v$.  This increases the number of parts of $\Sigma$ which we
must consider and could cause problems when we try to bound $|X^*_v|$.
But fortunately this always comes from a decrease in the genus of the
surface and the required inequality remains true.

If $\mu$ is a partition of $E_\Lambda$, then we denote by
$\delta_\Lambda^*(\mu)$ the set of faces of $\Lambda$ which are
incident with edges in at least two parts of $\mu$.

\begin{lemma}\label{lemma:contract3}
  Let $\Lambda$ be a 2-cell embedding of a hypergraph on a surface
  $\Sigma$ with at least three edges, let $\Pi$ be a radial embedding
  of $\Lambda$, and let $\mu=\{A, B, C\}$ be a $\Pi$-connected
  partition of $E_\Lambda$.

  There exists a 2-cell embedding $\Lambda'$ on a surface $\Sigma'$,
  and a partition $\mu'=\{A', B', C'\}$ of $E_{\Lambda'}$ such that
  \begin{enumerate}[i.]
  \item $V_{\Lambda'}=\delta_\Lambda(\mu)$, $F_{\Lambda'}=
    \delta^*_\Lambda(\mu)= \delta^*_{\Lambda'}(\mu')$;
  \item $k_{\Sigma'}\leq k_\Sigma+3-|E_{\Lambda'}|$.
  \end{enumerate}
\end{lemma}
\begin{proof}
  Since $A_\Pi$, $B_\Pi$ and $C_\Pi$ are disjoint, we can work
  independently in each one of them.  We thus start with $A_\Pi$.

  First we claim that we may assume that $\Lambda\cap A_\Pi$ is a
  connected subset of $A_\Pi$.  Since $A$ is $\Pi$-connected, it
  corresponds to a connected subgraph of $G^\Pi$.  If all the pairs
  $(e, f)\in A^2$ of $\Pi$-connected edges are such that
  $\overline{e\cup f}\cap A_\Pi$ is connected, then we are done.  So
  let $(e, f)\in A^2$ $\Pi$-connected edges with $\overline{e\cup
    f}\cap A_\Pi$ disconnected.  This can happen if $e$ and $f$ are
  incident with a vertex $v$ in $\delta_\Lambda(A)$.  In this case,
  let $D$ be a small disc around $v$.  We add a new vertex $v_e$ in
  $e\cap D$, a new vertex $v_f$ in $f\cap D$, and an edge $v_ev_f$ in
  $D\cap A_\Pi$.  When doing so, we split a face of $\Lambda$ which
  may belong to $\delta^*_\Lambda(A)$ but the new triangle face
  $vv_ev_f$ is only incident with edges in $\Pi_A$, and the other face
  belongs to $\delta^*_\Lambda(A)$ if and only if the original face
  did.  The claim follows.

  Let $\Gamma$ be the incidence embedding of $\Lambda$.  Since
  $\Lambda\cap A_\Pi$ is connected, we can contract a spanning tree of
  $\Gamma\cap A_\Pi$ so that it contains a single vertex $v_A$ linked
  to $\delta_\Lambda(A)$ together with some loops.  We then remove all
  the loops.  Let $e$ be such a loop.  Two cases arise:
  \begin{itemize}
  \item The loop $e$ bounds one face $F$ whose boundary is
    $v_AP_1v_Aev_AP_2$.  We remove $F\cup e$ from $\Sigma$ and add a
    disc whose boundary is $v_AP_1v_AP_2$.  If $F\cup e$ was a disc,
    then $e$ would split this disc in two parts which is not the case.
    The genus of the new surface is thus lower than the genus of the
    previous one.

  \item The loop $e$ bounds two faces $F_1$ and $F_2$ whose respective
    boundaries are $v_AP_1v_Ae$ and $v_AP_2v_Ae$.  Note that
    $e\cup\{v_A\}$ may bound an empty disc in $\Sigma$ in which case
    $P_1$ or $P_2$ is empty.  If either $F_1$ or $F_2$ does not belong
    to $\delta^*_\Lambda(A)$, then we remove $e$ and merge $F_1$ and
    $F_2$.  Otherwise, $e$ bounds no disc.  In this case, we remove
    $e$ from $\Lambda$ and cut $\Sigma$ along $e\cup\{v_A\}$.  More
    precisely, we remove $e\cup\{v_A\}$ from $\Sigma$ and take the
    closure of the resulting topological space.  We thus obtain a
    surface of lower genus and whose boundary is made of two holes
    $v_A^1e^1$ and $v_A^2e^2$.  We then fill the holes with two discs.
    Note that in the process, $v_A$ has been split in $v_A^1$ and
    $v_A^2$ but, as already mentioned the genus of the surface has
    dropped.
  \end{itemize}

  In the end, we have removed all the faces which were only incident
  with edges in $A$, and kept all the others.  We have replaced $A$
  by, say, $p$ edges and done some surgery on $\Sigma$ to keep the
  embedding 2-cell and the surgery resulted in a decrease of genus of
  at least $p-1$.  The lemma follows.
\end{proof}

We are now ready to prove:

\begin{theorem}\label{th:2}
  Let $\Lambda$ be a 2-cell embedding of a hypergraph with at least
  one edge on a surface $\Sigma$ and let $\Pi$ be a radial embedding
  of $\Lambda$.

  For any p-tree $T$ of $(\Lambda, \Pi)$,
  \begin{equation*}
    \tw(T^*)\leq\max\{\tw(T)+1+k_\Sigma, \alpha_{\Lambda^*}-1\}.
  \end{equation*}
\end{theorem}
\begin{proof}
  Let $v$ be a vertex of $T$, let $X_v$ be its bag in $T$ and let
  $X^*_v$ be its bag in $T^*$.  If $v$ is a leaf labelled by an edge
  $e$, then $X^*_v = e^*$ and $|X^*_v|-1\leq\max\{\tw(T)+1+k_\Sigma,
  \alpha_{\Lambda^*}-1\}$.  If $v$ is the neighbour of a leaf labelled
  by a troublesome edge $e$, then the fact that $\lambda_v$ is the
  $e$-partition implies that $X^*_v \subseteq e^*$ and
  $|X^*_v|-1\leq\max\{\tw(T)+1+k_\Sigma, \alpha_{\Lambda^*}-1\}$.

  We can thus suppose that $v$ is an internal node of $T$ whose
  node-partition $\lambda_v=\{A, B, C\}$ is $\Pi$-connected.  We then
  have $X_v=\delta_\Lambda(\lambda_v)$, and
  $X^*_v=\delta^*_\Lambda(\lambda_v)$.  Let $\Lambda'$ be given by
  Lemma~\ref{lemma:contract3} with $\mu=\lambda_v$.  Let $\Gamma$ be
  the incidence graph of $\Lambda'$.  We claim that any face of
  $\Gamma$ is incident with at least 4 edges.  This follows from the
  faces that $\Gamma$ is bipartite, and that no face of $\Gamma$ is
  incident with only one part $A$, $B$ and $C$.  Let $F_{2k}$ denotes
  the set of faces of length $2k$.  If an edge is incident with only
  one face $F$, then it counts twice in the length of $F$.  We have
  $2|E_\Gamma| = 4|F_4|+6|F_6|+\dots\geq 4|F_\Gamma|$, and thus since
  the face set of $\Gamma$ is exactly
  $\delta^*_\Lambda(\lambda_v)=X^*_v$,
  \begin{equation}
    \label{eq:E_leq_2F}
    |E_\Gamma| \geq 2|X^*_v|.
  \end{equation}

  Since $\Gamma$ has $|X^*_v|$ faces, $|X_v|+|E_{\Lambda'}|$ vertices,
  and since $k_{\Sigma'}\leq k_\Sigma+3-|E_{\Lambda'}|$, by replacing
  these in Euler's formula, we obtain:
  \begin{equation}
    \label{eq:euler}
    |X_v|+|E_{\Lambda'}|-|E_\Gamma|+|X^*_v|\geq 2-k_\Sigma-3+|E_{\Lambda'}|.
  \end{equation}

  Adding (\ref{eq:E_leq_2F}) and (\ref{eq:euler}), we get
  \begin{equation*}
    |X_v|+1+k_\Sigma\geq |X^*_v|
  \end{equation*}
  which proves that $|X^*_v|-1\leq\max\{\tw(T)+1+k_\Sigma,
  \alpha_{\Lambda^*}-1\}$, and thus
  $\tw(T^*)\leq\max\{\tw(T)+1+k_\Sigma, \alpha_{\Lambda^*}-1\}$.
\end{proof}

Theorem~\ref{thm}, which we restate below, is a direct corollary of
Theorem~\ref{th:existsT} and Theorem~\ref{th:2}.  \newcounter{toto}
\setcounter{toto}{\thetheorem} \setcounter{theorem}{0}
\begin{theorem}
  For any 2-cell embedding of a hypergraph $\Lambda$ on a surface
  $\Sigma$,
  \begin{equation*}
    \tw(\Lambda^*)\leq\max\{\tw(\Lambda)+1+k_\Sigma,
    \alpha_{\Lambda^*}-1\}.
  \end{equation*}
\end{theorem}

\setcounter{theorem}{\thetoto}

\section{Examples of graphs attaining the bound}
\label{sec:examples}

In Theorem~\ref{thm}, the $\alpha_{\Lambda^*}-1$ part is clearly
necessary.  Indeed, let $\Lambda^*$ be the plane embedding of the
hypergraph with one edge that contains $k$ vertices.  The dual
$\Lambda$ contains exactly one vertex, and thus $\tw(\Lambda)=0$ and
$k-1=\tw(\Lambda^*)=\alpha(\Lambda^*)-1$.  We now focus on embedding
of graphs.

A graph $G$ is \emph{minimally embeddable} on a surface $\Sigma$ if
there exists an embedding of $G$ in $\Sigma$ and for every embedding
of $G$ on a surface $\Sigma'$, then $k_{\Sigma'}\geq k_\Sigma$.  We
also say that $\Gamma$ is a minimum genus embedding on $\Sigma$.  In
this section, we prove the following theorem: \newcounter{titi}
\setcounter{titi}{\thetheorem}
\begin{theorem}\label{th:examples}
  For any surface $\Sigma$, there exists a minimum genus embedding
  $\Gamma$ on $\Sigma$ such that
  $\tw(\Gamma)=\tw(\Gamma^*)+1+k_\Sigma$.
\end{theorem}

To do so, we need some more definitions.  A \emph{bramble} of $G$ is a
family $\mathcal{B}$ of subsets of vertices of $G$ such that for every
element $X\in\mathcal{B}$, $G[X]$ is connected, and for all elements
$X$, $Y\in\mathcal{B}$, $G[X\cup Y]$ is connected (we say that $X$ and
$Y$ \emph{touch}).  The \emph{order} of a bramble $\mathcal{B}$ is the
minimum size of a set $X\subseteq V_G$ which intersects all the
elements of $\mathcal{B}$.  We use brambles to compute tree-width with
the following theorem:
\begin{theorem}[\cite{SeTh93a}]\label{th:bramble}
  The maximum order of a bramble of $G$ is $\tw(G)+1$.
\end{theorem}

Let $\Gamma$ be an embedding of a graph on a surface $\Sigma$ which is
not the sphere, and let $\theta>0$ be an integer.  We say that
$\Gamma$ is \emph{$\theta$-representative} if for every $\mu\subseteq
\Sigma$ homeomorphic to a circle, if $\gamma$ is not the boundary of
some closed disc in $\Sigma$ (we call such $\mu$ a
\emph{non-contractile noose}) then $|\gamma\cap \Gamma|\geq\theta$.
We use representativity to certify that graphs are minimally embedded
with the following theorem which is an easy corollary of a theorem
in~\cite{SeTh96a}.
\begin{theorem}\label{th:representativity}
  Let $\Gamma$ be a $\theta$-representative embedding of a graph on a
  surface $\Sigma$ which is not the sphere.  If $\theta\geq
  100k_\Sigma$ then $\Gamma$ is a minimum genus embedding.
\end{theorem}

\subsection{Todinca graphs and their tree-widths}

Let $p\geq 1$ be an integer.  Let $A$, $B$ and $C$ be three $2p\times
2p$ grids.  Let $a_1$,~\dots, $a_p$, $a'_p$,~\dots, $a'_1$, be the top
row of $A$.  Let $b_1$,~\dots, $b_p$, $b'_p$, ~\dots, $b'_1$ be the
top row of $B$.  Let $c_1$,~\dots, $c_p$, $c'_p$,~\dots, $c'_1$ be the
top row of $C$.  A \emph{Todinca graph\footnote{The name come from
    Ioan Todinca who first showed us the plane graph of
    Figure~\ref{fig:TodG1} as an example of graph whose
    $\text{tree-width}+1$ equals $3/2$ its branch-width.}  of order
  $p$} is any graph $G$ obtained by bijectively linking the edges
between $a_1$,~\dots, $a_p$ and $b'_1$,~\dots, $b'_p$, between
$b_1$,~\dots, $b_p$ and $c'_1$,~\dots, $c'_p$ and between $c_1$,
~\dots, $c_p$ and $a'_1$,~\dots, $a_p$.
\begin{figure}[htbp]
  \centering

  \includegraphics{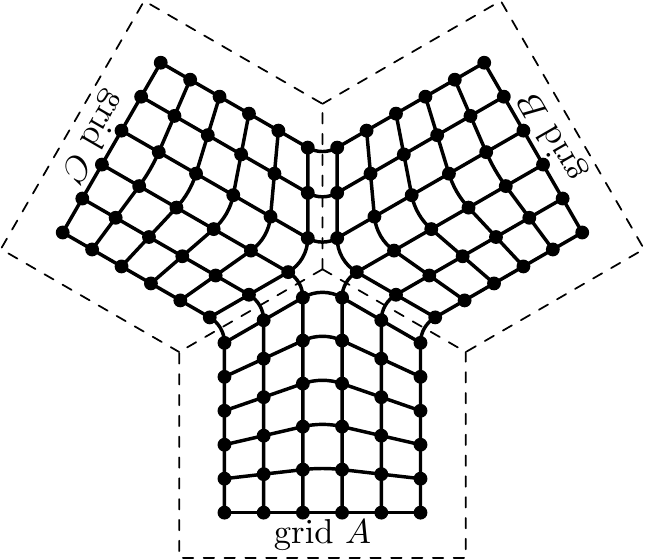}\quad
  \includegraphics{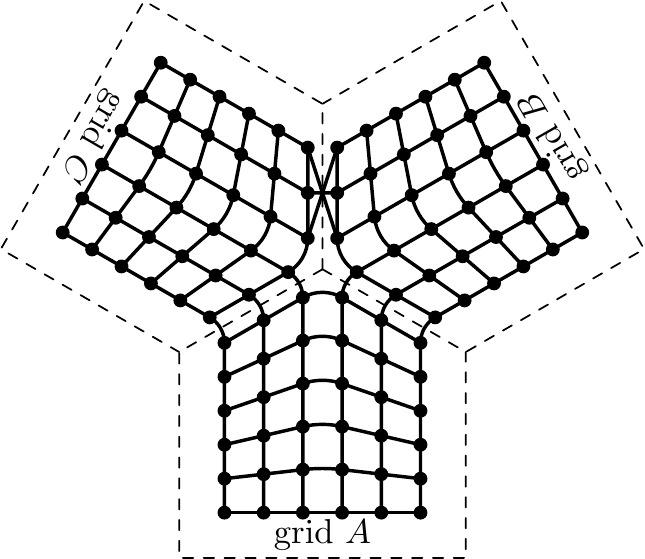}

  \caption{Two Todinca graphs.}
  \label{fig:TodG1}
\end{figure}
As an example, in Figure~\ref{fig:TodG1}, the left graph is obtained
by adding the edges $a_ib'_i$, $b_ic'_i$ and $c_ia'_i$, while in the
right graph, we link $b_i$ to $c'_{p+1-i}$.

\begin{lemma}\label{lem:bramble}
  The tree-width of every Todinca graph $G$ of order $p$ is at least
  $3p-1$.
\end{lemma}
\begin{proof}
  By Theorem~\ref{th:bramble}, we only have to produce a bramble of
  order $3p$ to prove our lemma.  To do so, we give some definitions.

  The edges added to the grids $A$, $B$ and $C$ link their columns and
  thus define \emph{columns} of $G$.  Depending on which columns were
  linked, we obtain \emph{$AB$-columns}, \emph{$BC$-columns} and
  \emph{$CA$-columns}.  We also call the respective rows of $A$, $B$
  and $C$, the \emph{$A$-rows}, the \emph{$B$-rows} and the
  \emph{$C$-rows} of $G$.  Together, they are the \emph{rows} of $G$
  (see Figure~\ref{fig:rows-columns}).
  \begin{figure}[htbp]
    \centering

    \includegraphics{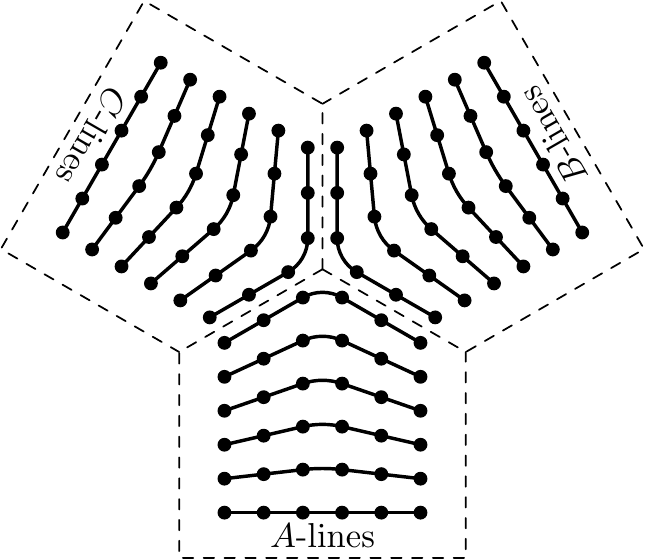}\quad
    \includegraphics{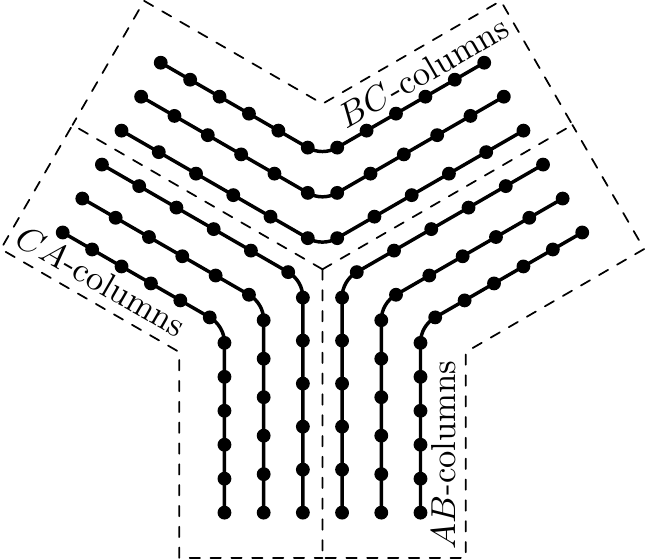}

    \caption{The rows and the columns of a Todinca graph.}
    \label{fig:rows-columns}
  \end{figure}
  The union of an $A$-row and an $AB$-column is an \emph{$A$-cross}.
  In the same spirit, we define \emph{$B$-crosses} and
  \emph{$C$-crosses}, and we claim that the set $\mathcal{C}$ of all
  these \emph{crosses} is a bramble of $G$ of order $3p$.

  They clearly are connected, so let us prove that any two crosses $X$
  and $Y$ touch.  By symmetry, we can suppose that $X$ is an
  $A$-cross.  Since $X$ contains an $A$-row, and since $A$-rows
  intersect all $AB$- and all $CA$-columns, if $Y$ is an $A$- or a
  $C$-cross, then $X$ and $Y$ touch.  If $Y$ is a $B$-cross, then it
  contains a $B$-row which intersects all $AB$-columns.  The crosses
  $X$ and $Y$ therefore touch, which finishes our proof that
  $\mathcal{C}$ is a bramble.

  Since there are exactly $3p$ columns, the order of $\mathcal{C}$ is
  at most $3p$.  Let us prove that no set $S$ of size $3p-1$ can
  intersect all the crosses.  Obviously, the set $S\subseteq V_G$
  misses at least one column.  We split this proof in two cases
  \begin{itemize}
  \item There exists an $AB$-column $X$, a $BC$-column $Y$ and a
    $CA$-column $Z$ avoiding $S$.  Since $G$ contains $6p$ rows, there
    exists a row $L$ avoiding $S$.  Depending on what kind of row $L$
    is, one of $L\cup X$, $L\cup Y$ and $L\cup Z$ is a cross of $G$
    avoiding $S$.
  \item There exists, by symmetry, an $AB$-column $C$ avoiding $S$,
    and no $BC$-column avoiding $S$.  Since there are at least $p$
    vertices of $S$ on the $BC$-columns, there can only be $2p-1$
    vertices from $S$ on the grid $A$.  There thus exists an $A$-row
    $L$ avoiding $S$.  The cross $L\cup C$ avoids $S$.
  \end{itemize}
\end{proof}

We also need the following folklore lemma (see
Figure~\ref{fig:tree-width_grid}).
\begin{lemma}\label{lem:tree-width_grid}
  An $n\times m$ grid $G$ has tree-width at most $\min(n, m)$, and $G$
  admits a path-decomposition attaining this bound with a leaf
  containing one of its shortest sides.
\end{lemma}
\begin{figure}[htbp]
  \centering

  \includegraphics{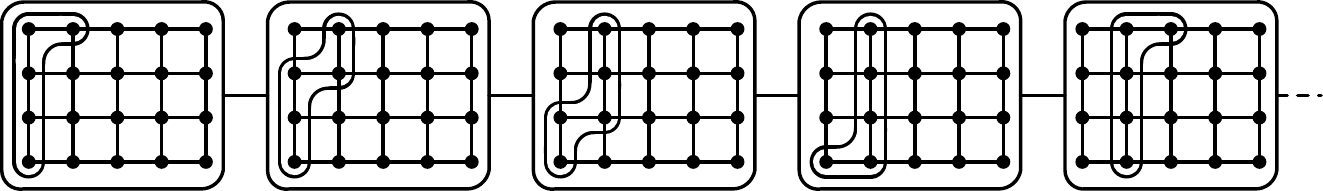}

  \caption{A tree-decomposition of the $4\times 5$ grid of width $4$.}
  \label{fig:tree-width_grid}
\end{figure}
As a matter of fact, the tree-width of an $n\times m$ grid $G$ is
exactly $\min(n, m)$ (see~\cite{BoGrKo08a}).

\begin{lemma}\label{lem:tw}
  The tree-width of every Todinca graph $G$ of order $p$ is at most
  $3p-1$.
\end{lemma}
\begin{proof}
  To prove or lemma, we give a tree-decomposition of $G$ of width
  $3p-1$.  Let $A$, $B$ and $C$ be the three $2p\times 2p$ grids of
  $G$, and let $\mathcal{T}_A$, $\mathcal{T}_B$ and $\mathcal{T}_C$ be
  path decompositions of $A$, $B$ and $C$ given by
  Lemma~\ref{lem:tree-width_grid}.  Let $v_A$, $v_B$ and $v_C$ be the
  vertices of $T_A$, $T_B$ and $T_C$ whose bag respectively contain
  the top rows of $A$, $B$ and $C$.  Let $u$, $u_A$, $u_B$ and $u_C$
  be vertices whose bags respectively are $\{a_1,~\dots, a_p,
  b_1,~\dots, b_p, c_1,~\dots, c_p\}$, $\{a_1,~\dots, a_p,
  a'_p,~\dots, a'_1, b_1,~\dots, b_p\}$, $\{b_1,~\dots, b_p,
  b'_p,~\dots, b'_1, c_1,~\dots, c_p\}$ and $\{c_1,~\dots, c_p,
  c'_p,~\dots, c'_1, a_1,~\dots, a_p\}$.  The labelled tree $T$
  obtained by linking $u$ to $u_A$, $u_B$ and $u_C$, and adding the
  edges $u_Av_A$, $u_Bv_B$ and $u_Cv_C$ is a tree-decomposition of $G$
  of width $3p-1$.
\end{proof}

Lemmata~\ref{lem:bramble} and~\ref{lem:tw} clearly imply that
\begin{lemma}\label{th:tw}
  The tree-width of every Todinca graph of order $p$ is $3p-1$.
\end{lemma}

\subsection{Some minimally embeddable Todinca graphs}

We now define special Todinca graphs.  We first define three gadgets
(see Figure~\ref{fig:gadgets}) that we use to link the grids of our
graphs.
\begin{figure}[htbp]
  \centering

  \subfloat[An
  $l$-ladder]{\includegraphics{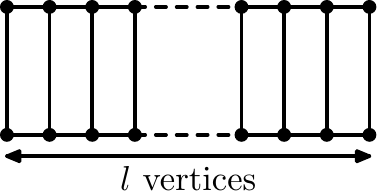}} \hfill%
  \subfloat[An
  $l$-handle]{\includegraphics{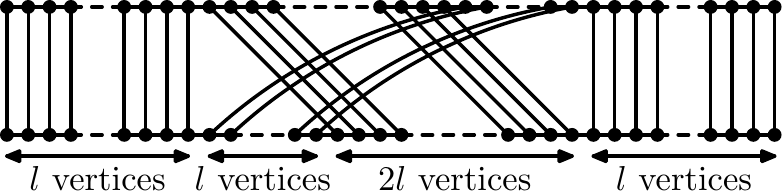}} \hfill%
  \subfloat[An
  $l$-crosscap]{\includegraphics{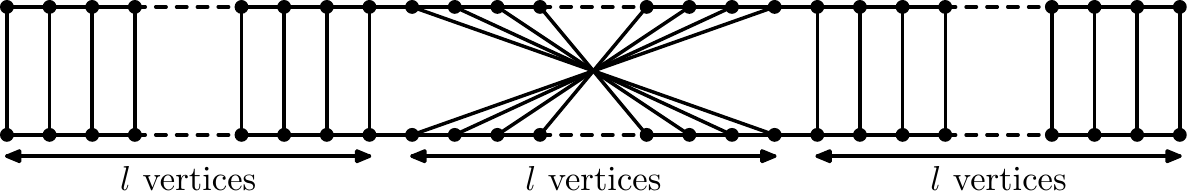}}

  \caption{linking gadgets}
  \label{fig:gadgets}
\end{figure}
Let $P$ and $Q$ be two paths of length $k$, and let $p_1$, \dots,
$p_k$ and $q_1$, \dots, $q_k$ be their respective vertices.  If we
link each $p_i$ with $q_i$ ($1\leq i\leq k$), we obtain a
\emph{$k$-ladder}.  If $k=3l$ and we link $p_i$ with $q_i$ ($1\leq
i\leq l$ and $2l+1\leq i\leq 3l$) and $p_{l+i}$ with $q_{2l+1-i}$
($1\leq i\leq l$), then we obtain an \emph{$l$-crosscap}.  If $k=5l$
and we link $p_i$ with $q_i$ ($1\leq i\leq l$ and $4l+1\leq i\leq
5l$), $p_{l+i}$ to $q_{2l+i}$ ($1\leq i\leq 2l$) and $p_{3l+i}$ to
$q_{l+i}$ $(1\leq i\leq l$), then we obtain an \emph{$l$-handle}.

Let $k>0$ and $p>1$ be integers.  The graph $G_{k, p}$ is the Todinca
graph of order $5kp$ such that the grid $A$ is linked to the grids $B$
and $C$ by a $5kp$-ladder, and the grids $B$ and $C$ are linked $k$
$p$-handles.  The graph $\widetilde G_{k, p}$ is the Todinca graph of
order $3kp$ defined in the same way except that we use $p$-crosscaps
instead of $p$-handles.

We now define an embedding $\Gamma_{k, p}$ of $G_{k, p}$ in $\Sigma_k$
as follows.  We start by embedding the three grids and the
$5kp$-ladders on the sphere.  We also embed the all the edges of the
handles except the edges $p_{3l+i}q_{l+i}$ $(1\leq i\leq l$).  Then
for each $p$-handle, we add a handle to the surface to embed the
remaining $p$ edges.  We also define an embedding $\widetilde
\Gamma_{k, p}$ of $\widetilde G_{k, p}$ in $\widetilde\Sigma_k$ in a
similar way except that for each $p$-crosscap, we add a crosscap to
the surface so that the edges $p_{l+i}q_{2l+1-i}$ ($1\leq i\leq l$) do
not cross.

We now want to prove that for $p$ large enough, these embeddings are
minimal genus embeddings.  To do so, we want to apply
Theorem~\ref{th:representativity}.  We thus have to prove that for $p$
large enough, these embeddings have large representativity.  More
precisely,

\begin{lemma}\label{lem:representativity}
  The embeddings $\Gamma_{k, p}$ and $\widetilde \Gamma_{k, p}$ are
  $p$-representative.
\end{lemma}
\begin{proof}
  The result for $\Gamma_{k, p}$ and $\widetilde \Gamma_{k, p}$ are
  respectively direct consequences of Theorems 3.5 and 3.3
  of~\cite{RoSe96a}, but very few people seem to have read the graph
  minors paper so we give a direct proof which is, in spirit, very
  similar to the proofs by Robertson and Seymour.

  We start with $\Gamma_{k, p}$.  As already said, to embed $G_{k, p}$
  on $\Sigma_k$, we start by embedding a subgraph of $G_{k, p}$ on the
  sphere, then for each handle gadget $P_i$ ($1\leq ieq k$), we add a
  handle to embed the edges $p_{3l+i}q_{l+i}$ $(1\leq i\leq l$).  To
  do so, we remove two faces $f_i$ and $g_i$ of our partial embedding
  and we ``sew'' a cylinder $H_i$ to the border of the holes.  It is
  easy to find $p$ vertex disjoint concentric cycle $c^i_1$,~\dots,
  $c^i_p$ enclosing $f_i$ and $p$ vertex disjoint concentric cycles
  $d^i_1$,~\dots, $d^i_p$ enclosing $g_i$ such that all the circles
  are pairwise disjoint (see left part of
  Figure~\ref{fig:concentric_cycles}).
  \begin{figure}[htbp]
    \centering

    \includegraphics{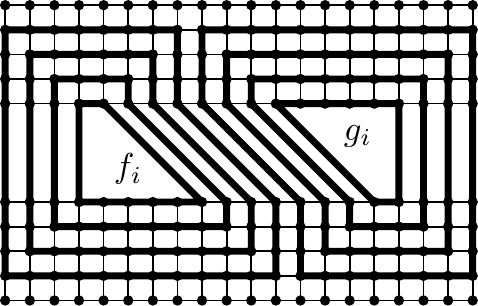}\qquad
    \includegraphics{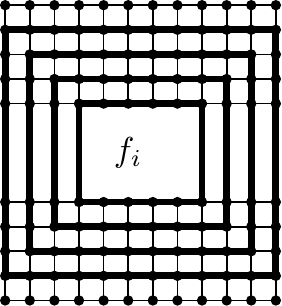}

    \caption{Enclosing concentric cycles around a handle or a
      crosscap.}
    \label{fig:concentric_cycles}
  \end{figure}
  The component of $\Sigma_k\setminus\bigl(\bigcup_{i=1}^k (c^i_p\cup
  d^i_p)\bigr)$ which is incident with all the cycles $c^i_p$ and
  $d^i_p$ is the \emph{outside region}.  Let $\mu$ be a
  non-contractile noose on $\Sigma_k$.  We split our proof in the
  following three sub-cases.
  \begin{itemize}
  \item Suppose that $\mu$ intersects both the outside region and the
    handle $H_i$.  The curve $\mu$ has to cross all the cycles
    $c^i_1$,~\dots, $c^i_p$ or all the cycles $d^i_1$,~\dots, $d^i_p$
    to reach the outside region, which implies that $|\mu\cap
    \Gamma_{k, p}|\geq p$.
  \item Suppose that $\mu$ intersects a handle $H_i$.  Then $\mu$ is a
    subset of the component $H'_i$ of $\Sigma_k\setminus (c^i_p\cup
    d^i_p)$ which contains $H_i$.  It is easy to find $p$ vertex
    distinct paths $Q_1$,~\dots, $Q_p$ in $\Gamma_{k, p}$ which link
    vertices in $c^i_p$ and vertices in $d^i_p$ and whose interior is
    in $H'_i$ (see left part of Figure~\ref{fig:inside_linkage}).
    \begin{figure}[htbp]
      \centering

      \includegraphics{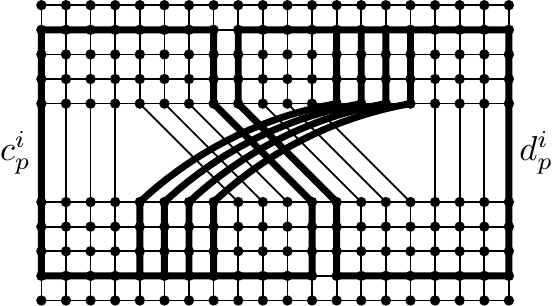}\qquad
      \includegraphics{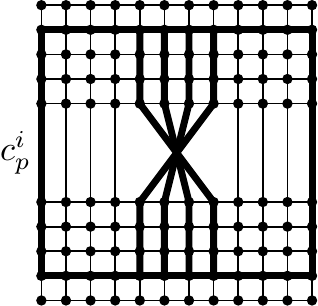}

      \caption{Paths across a handle or a crosscap.}
      \label{fig:inside_linkage}
    \end{figure}

    We claim that $\mu$ intersects all the paths $Q_i$.  Indeed
    otherwise $\mu$ is a subset of $H'_i\setminus Q_j$ for some $1\leq
    j\leq p$.  But $H'_i\setminus Q_j$ is an open disc which implies
    that $\mu$ is contractible, a contradiction.  This thus implies
    that $|\mu\cap \Gamma_{k, p}|\geq p$.

  \item Suppose that $\mu$ intersects the outside region.  Then $\mu$
    is a subset of
    $\Sigma'=\Sigma_k\setminus\bigl(\bigcup_{i=1}^kH_i\bigr)$.  Since
    $\Sigma'$ is a sphere with holes bounded by the cycles $c^i_1$ and
    $d^i_1$, then $\mu$ separates $c^1_1$ and $c^i_1$ or $d^i_1$ in
    $\Sigma'$.
    \begin{figure}[htbp]
      \centering

      \includegraphics{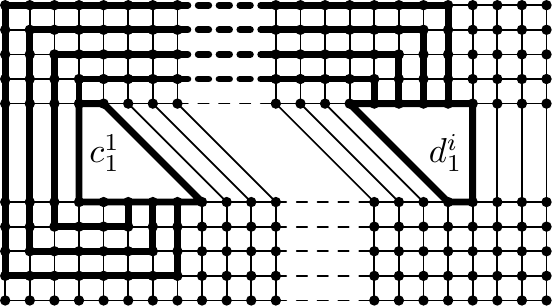}\quad
      \includegraphics{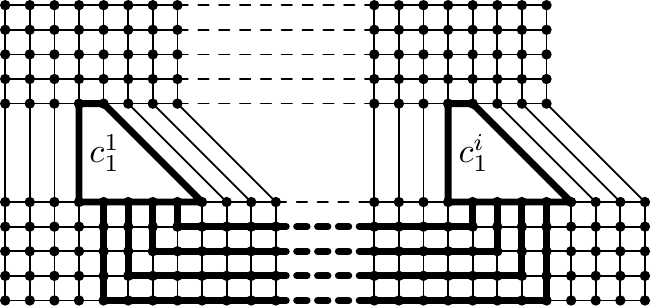}\bigskip

      \includegraphics{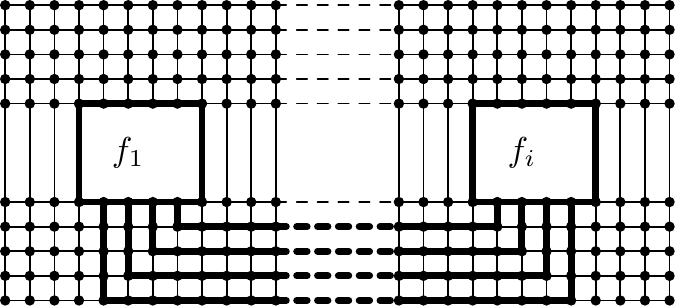}

      \caption{Paths linking handles or crosscaps in the outside
        region.}
      \label{fig:linkage_externe}
    \end{figure}
    In both cases, it is easy to find $p$ vertex disjoint paths in
    $\Gamma_{k, p}$ between $c^1_1$ and $c^i_1$ or $d^i_1$ whose
    interior are in $\Sigma'$ (see upper part of
    Figure~\ref{fig:linkage_externe}).  Since $\mu$ has to cut these
    paths, $|\mu\cap \Gamma_{k, p}|\geq p$.
  \end{itemize}

  In all cases, $|\mu\cap \Gamma_{k, p}|\geq p$ which finishes our
  proof that the embedding of $\Gamma_{k, p}$ is
  $p$-representative.\medskip

  The proof that $\widetilde \Gamma_{k, p}$ if $p$-representative is
  very similar.  We first enclose each crosscap by $p$ disjoint
  cycles, and the we define an outside region.  We then prove that any
  noose $\mu$ which intersects both the outside and a crosscap meets
  the embedding at least $p$ times.  If $\mu$ intersects a crosscap,
  it is enclosed by the outer cycle $c$ enclosing the crosscap.  We
  can easily find $p$ disjoint paths each linking two points of $c$ as
  in the right part of Figure~\ref{fig:inside_linkage}, and $\mu$ has
  to meet all those paths otherwise it is contractible.  If $\mu$
  intersects the outer region, then it must separate some crosscaps
  and since it is easy to find $p$ disjoint path linking the inner
  cycles enclosing these crosscaps, $\mu$ intersects the embedding at
  least $p$ times.
\end{proof}

As a consequence of Lemma~\ref{lem:representativity} and
Theorem~\ref{th:representativity}, we have
\begin{lemma}\label{lem:min-genus-embeddings}
  $\Gamma_{k, 100k}$ and $\widetilde \Gamma_{k, 100k}$ are minimum
  genus embeddings respectively in $\Sigma_k$ and
  $\widetilde\Sigma_k$.
\end{lemma}

\subsection{Dual of some Todinca graphs}

\begin{lemma}\label{lem:tw-dual-orientable}
  The tree-width of $\Gamma_{k, p}^*$ is at most $15kp-2-k$.
\end{lemma}
\begin{proof}
  The embedding $\Gamma_{k, p}^*$ is made of three
  $(10kp-1)\times(10kp-1)$ grids $A$, $B$ and $C$, two paths of
  $5kp-1$ vertices $P_{AB}$ and $P_{AC}$, a graph $G_{BC}$ which
  corresponds to the dual of the gadgets and two vertices
  $v_\text{in}$ and $v_\text{out}$.  The path $P_{AB}$ is adjacent to
  the grids $A$ and $B$, the path $P_{AC}$ is adjacent to the grids
  $A$ and $C$, and the graph $G_{BC}$ is adjacent to the grids $B$ and
  $C$.  The vertex $v_\text{out}$ is adjacent to the vertices of the
  bottom rows of $A$, $B$ and $C$, to the vertices of the side columns
  of $A$, $B$ and $C$ and to the ``outer vertex'' of $P_{AB}$,
  $P_{AC}$ and $G_{BC}$.  The vertex $v_\text{in}$ is adjacent to the
  middle vertex of the top rows of $A$, $B$, $C$ and to the ``inner
  vertex'' of $P_{AB}$, $P_{AC}$ and $G_{BC}$ (see
  Figure~\ref{fig:dual}).

  \begin{figure}[htbp]
    \centering
    \includegraphics{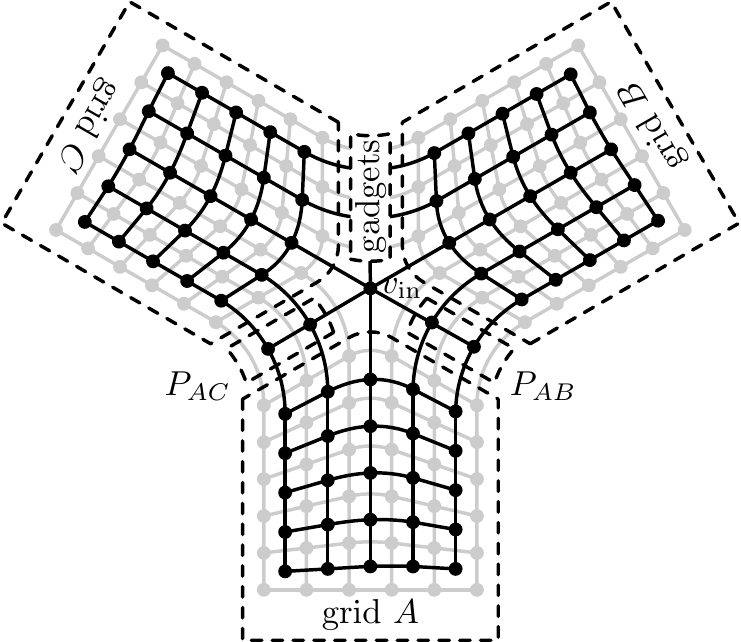}
    \caption{The embedding $\Gamma^*_{k, p}-v_\text{out}$.}
    \label{fig:dual}
  \end{figure}

  $\Gamma_{k, p}$ is an embedding of a Todinca graph of order $l=5kp$
  in a surface of Euler genus $2k$.  It thus has $12l^2$ vertices and
  $24l^2-9l$ edges.  By Euler's formula, it thus has $12l^2-9l+2-2k$
  faces.  This number is also the number of vertices of $\Gamma^*_{k,
    p}$.  There are $3(2l-1)^2$ vertices in the grids, $2(l-1)$
  vertices on $P_{AB}$ and $P_{AC}$ and the two vertices $v\text{in}$
  and $v_\text{out}$.  This leaves $l-2k-1$ \emph{gadget vertices} on
  $G_{BC}$.  Since the tree-width of an embedding drops by at most one
  when removing a single vertex, we only have to prove that the
  tree-width of $\Gamma^*_{k, p}-v_\text{out}$ is $15kp-3-2k$ so let
  us remove the vertex $v_\text{out}$ from $\Gamma^*_{k, p}$.

  The grid $A$ together with its neighbourhood is a $2l\times (2l-1)$
  grid.  We can thus choose a path decomposition $\mathcal{T}_A$ of
  width $2l-1$ of this grid and in which a vertex $v_A$ contains the
  neighbourhood of $A$.  Because of the gadgets, the links between the
  grid $B$ with its neighbourhood is more complex.  Let $v_1$,~\dots,
  $v_{2l-1}$ be the vertices of the top row of $B$ which link $B$ to
  the remaining of $\Gamma^*_{k, p}$.  The vertices $v_1$,~\dots,
  $v_l$ are clearly linked to $P_{AB}\cup\{v_\text{in}\}$ in an
  $l$-ladder.  Since the gadget sequence begins with a $p$-ladder,
  there is a set $S$ of $p-1$ gadget vertices such that $v_1$,~\dots,
  $v_{l+p-1}$ is linked to $P_{AB}\cup\{v_\text{in}\}\cup S$ in a
  $l+p-1$ ladder.  The remaining gadget vertices are linked to the
  vertices $v_{l+p}$,~\dots, $v_{2l-1}$.  Using the idea behind the
  path decomposition of Figure~\ref{fig:tree-width_grid}, we define a
  path decomposition $\mathcal{T}_B$ of width $3l-p-2k-2$ of $B$ and
  its neighbourhood in which a vertex $v_B$ contains $P_{AB}$,
  $v_\text{in}$, the gadget vertices and $v_{l+p}$,~\dots, $v_{2l-1}$.
  We can similarly define a tree-decomposition $\mathcal{T}_C$ of $C$
  and its neighbourhood in which a vertex $v_C$ contains the
  neighbourhood of $C$.  Let $u$ be a vertex whose bag contains
  $P_{AB}$, $P_{AC}$ $v_\text{in}$ and the gadget vertices, and let us
  add the edges $uv_A$, $uv_B$ and $uv_C$.  This defines a
  tree-decomposition of $\Gamma^*_{k, p}-v_\text{out}$.  The bag of
  $u$ is its biggest one and it contains $3l-2-2k$ vertices which
  proves that the tree-width of $\Gamma^*_{k, p}$ is at most
  $3l-2k-2=15kp-2-2k$.
\end{proof}

Using a similar proof, we also obtain:
\begin{lemma}
  The tree-width of $\widetilde\Gamma_{k, p}^*$ is at most $9kp-2-k$.
\end{lemma}

We are now ready to prove Theorem~\ref{th:examples}.
\setcounter{theorem}{\thetiti}
\begin{theorem}
  For any surface $\Sigma$, there exists a minimum genus embedding
  $\Gamma$ on $\Sigma$ such that $\tw(\Gamma) =
  \tw(\Gamma^*)+1+k_\Sigma$.
\end{theorem}
\begin{proof}
  In~\cite{RoSe84a}, Robertson and Seymour already gave examples of
  planar embeddings matching our bound.  So let us consider higher
  genus surfaces.

  $\Gamma_{k, 100k}$ is an embedding of a Todinca graph of order
  $500k^2$.  By Lemma~\ref{th:tw}, its tree-width is $1500k^2-1$.  By
  Lemma~\ref{lem:tw-dual-orientable}, the tree-width of $\Gamma_{k,
    100k}^*$ is at most $1500k^2-2-2k$.  Since $\Gamma_{k, 100k}$ is a
  an embedding in $\Sigma_k$, Theorem~\ref{thm} implies that
  $\tw(\Gamma_{k, 100k}^*)=1500k^2-2-2k$ and $\tw(\Gamma_{k,
    100k})-\tw(\Gamma_{k, 100k}^*)=1+2k$.  Since, by
  Lemma~\ref{lem:min-genus-embeddings}, $\Gamma_{k, 100k}$ is a
  minimum genus embedding in $\Sigma_k$, $\Gamma_{k, 100k}$ and
  $\Gamma^*_{k, 100k}$ indeed are examples of embeddings matching the
  bound of Theorem~\ref{thm} for the surface $\Sigma_k$.

  Similarly, $\widetilde\Gamma_{k, 100k}$ and $\widetilde\Gamma^*_{k,
    100k}$ are examples of embeddings matching the bound of
  Theorem~\ref{thm} for the surface $\widetilde\Sigma_k$.
\end{proof}

\section{Conclusion and open questions}

In this paper, we show that tree-width is a robust parameter
considering surface duality.  Indeed, our main proof says more than
just ``the difference between the tree-width of $\Lambda$ and that of
$\Lambda^*$ is small''.  Our proof says that there always exists a
decomposition which is optimal for $\Lambda$ and very good for
$\Lambda^*$.  This leads to a natural question: For any embedding
$\Lambda$, does there always exists a p-tree $T$ such that
$\tw(T)=\tw(\Lambda)$ and $\tw(T^*)=\tw(\Lambda^*)$?  To our
knowledge, the question is open, even for plane embeddings of graphs.


\begin{thebibliography}{BMT03}
\bibitem[AP89]{ArPr89a}%
  Stephan Arnborg and Andrzej Proskurowski.%
  \newblock Linear time algorithms for {NP}-hard problems restricted
  to partial $k$-trees.%
  \newblock {\em Discrete Applied Mathematics}, 23:11--24, 1989.

\bibitem[BGK08]{BoGrKo08a}%
  Hans~Leo Bodlaender, Alexander Grigoriev, and Arie Marinus
  Catharinus~Antonius Koster.%
  \newblock Treewidth lower bounds with brambles.%
  \newblock {\em Algorithmica}, 51:81--98, 2008.

\bibitem[BMT03]{BoMaTo03a}%
  Vincent Bouchitt{\'e}, Fr{\'e}d{\'e}ric Mazoit, and Ioan Todinca.%
  \newblock Chordal embeddings of planar graphs.%
  \newblock {\em Discrete Mathematics}, 273:85--102, 2003.

\bibitem[Cou90]{Co90a}%
  Bruno Courcelle.%
  \newblock The monadic second-order logic of graphs I: recognizable
  sets of finite graphs.%
  \newblock {\em Information and Computation}, 85(1):12--75, 1990.

\bibitem[Hal76]{Ha76a}%
Rudolf Halin.%
\newblock {$S$}-functions for graphs.%
\newblock {\em Journal of Geometry}, 8:171--186, 1976.

\bibitem[Lap96]{La96b}%
  Denis Lapoire.%
  \newblock Treewidth and duality for planar hypergraphs.%
  \newblock Manuscript
  \url{http://www.labri.fr/perso/lapoire/papers/dual_planar_treewidth.ps},
  1996.

\bibitem[Maz04]{Ma04a}%
  Fr{\'e}d{\'e}ric Mazoit.%
  \newblock {\em D{\'e}composition algorithmique des graphes}.%
  \newblock PhD thesis, {\'E}cole normale sup{\'e}rieure de Lyon,
  2004.

\bibitem[Maz09]{Ma09b}%
  Fr{\'e}d{\'e}ric Mazoit.%
  \newblock Tree-width of graphs and surface duality.%
  \newblock {\em Electronic Notes in Discrete Mathematics}, 32:93--97,
  2009.


\bibitem[RS84]{RoSe84a}%
  Neil Robertson and Paul~D. Seymour.%
  \newblock Graph Minors. III. Planar Tree-Width.%
  \newblock {\em Journal of Combinatorial Theory Series B},
  36(1):49--64, 1984.

\bibitem[RS96]{RoSe96a}%
  Neil Robertson and Paul~D. Seymour.%
  \newblock Graph Minors. XV. Giant Steps.%
  \newblock {\em Journal of Combinatorial Theory Series B},
  68(1):112--148, 1996.

\bibitem[ST93]{SeTh93a}%
  Paul~D. Seymour and Robin Thomas.%
  \newblock Graph Searching and a Min-Max Theorem for Tree-Width.%
  \newblock {\em Journal of Combinatorial Theory Series B},
  58(1):22--33, 1993.

\bibitem[ST96]{SeTh96a}%
  Paul~D. Seymour and Robin Thomas.%
  \newblock Uniqueness of Highly Representative Surface Embeddings.%
  \newblock {\em Journal of Graph Theory}, 23(4):337--349, 1996.

\end{thebibliography}
\end{document}